\documentclass[11pt]{article}
\usepackage{geometry}
\usepackage{fullpage}
\usepackage{times}
\usepackage{graphicx}
\usepackage{amsmath}

\usepackage{amssymb}
\usepackage{amsthm}
\usepackage{epstopdf}
\newtheorem{thm}{Theorem}[section]
\newtheorem{lem}[thm]{Lemma}

\newtheorem{defn}[thm]{Definition}
\newtheorem{rem}[thm]{Remark}

\newcommand{\E}{\textrm{\textbf{E}}}
\newcommand{\mcA}{\mathcal{A}}
\newcommand{\valpha}{\vec{\alpha}}

\usepackage{hyperref}
\hypersetup{
bookmarksnumbered
}

\begin{document}
\title{An Oblivious $O(1)$-Approximation for Single Source
Buy-at-Bulk}
\author{Ashish Goel
\thanks{Departments of Management Science and Engineering, and by courtesy,
Computer Science, Stanford University. Email: ashishg@stanford.edu.
Research supported by an NSF ITR grant and the Stanford-KAUST alliance for academic excellence.}
\\Stanford University
\and Ian Post
\thanks{Department of Computer Science, Stanford University. Email: itp@stanford.edu.
Research supported by an NSF ITR grant and the Stanford-KAUST alliance for academic excellence.}
\\Stanford University}

\maketitle
\thispagestyle{empty}
\pdfbookmark[1]{Abstract}{MyAbstract}
\begin{abstract}
We consider the single-source (or single-sink) buy-at-bulk problem with an
unknown concave cost function.  We want to route a set of demands along a
graph to or from a designated root node, and the cost of routing $x$ units
of flow along an edge is proportional to some concave, non-decreasing
function $f$ such that $f(0) = 0$.  We present a polynomial time algorithm
that finds a distribution over trees such that the expected cost of a tree
for any $f$ is within an $O(1)$-factor of the optimum cost for that $f$. The
previous best simultaneous approximation for this problem, even ignoring
computation time, was $O(\log |\mathcal{D}|)$, where $\mathcal{D}$ is the
multi-set of demand nodes.

We design a simple algorithmic framework using the ellipsoid method that
finds an $O(1)$-approximation if one exists, and then construct a separation
oracle using a novel adaptation of the Guha, Meyerson, and Munagala
\cite{guha2001cfa} algorithm for the single-sink buy-at-bulk problem that
proves an $O(1)$ approximation is possible for all $f$. The number of trees
in the support of the distribution constructed by our algorithm is at most
$1+\log |{\mathcal D}|$.
\end{abstract}
\newpage

\section{Introduction}

We study the single-source (or single-sink) buy-at-bulk network design problem with an unknown concave cost function. 
We are given
an undirected graph $G=(V,E)$ with edge lengths $l_e$ and a set of
demand nodes $\mathcal{D} \subseteq V$ with integer demands $d_v$ and 
want to route these demands to a
designated root node $r$ as cheaply as possible, where the cost of routing
along a particular edge is proportional to some function $f$ of the amount of
flow sent along the edge. 
In many applications it is natural to assume that
$f$ is a concave, non-decreasing function such that $f(0) = 0$, capturing the
case where we benefit from some kind of economy of scale when aggregating
flows together.  We call such functions \emph{aggregation functions} and
define $\mathcal{F}$ as the set of all aggregation functions.

When the function $f$ is given, the problem becomes the well-studied
single-sink buy-at-bulk (SSBaB) problem.  SSBaB is $NP$-hard, since it
contains the Steiner tree problem as a special case.  The problem was
introduced by Salman et al.~\cite{salman1997bbn} who gave algorithms for
special cases.  Awerbuch and Azar \cite{awerbuch1997bbn} gave an $O(\log^2
n)$-approximation using metric tree embedding, which subsequently improved to
$O(\log n)$ using better metric embeddings~\cite{bartal1998aam,
fakcharoenphol2003tba}.  Building on their own work on hierarchical facility
location~\cite{guha2000hpa}, Guha, Meyerson, and Munagala (GMM) gave the first
constant-factor approximation \cite{guha2001cfa}, an algorithm that features
prominently in our results.  Recent work \cite{talwar2002ssb, gupta2003sab,
jothi2004iaa, grandoni2006ias} has reduced the approximation ratio to 24.92
and also provided an elegant cost-sharing framework for thinking about this
problem.

However, for some applications we may want to assume that $f$ is unknown or is
known to vary over time. For instance, we may be aggregating observations in a
sensor network where we do not know the amount of redundancy among different
observations or where the redundancy is known to change. In this setting, it
is desirable to find a solution that is robust to changes in $f$ and provides
a constant-factor approximation {\em simultaneously} for all $f \in
\mathcal{F}$.  Moreover, from a purely theoretical perspective, the existence
of a good algorithm that is independent of $f$ reveals non-trivial structure
in the problem.

We will focus on randomized algorithms.
Given the concavity of $f$, we may assume without loss of
generality that the optimal routing graph is a tree.  Let $\mathcal{T}$ be the
set of all trees in $G$ spanning $\mathcal{D}$ and $r$, and let $T_f^*$ be the
optimal tree for some fixed $f$.  We use the shorthand $f(T)$ to denote the
cost of $T$ under $f$, i.e.\ $\sum_e l_ef(x_{T,e})$ where $x_{T,e}$ is the
amount of flow tree $T$ routes on edge $e$.  
There are two natural objectives which capture simultaneous approximation for multiple cost functions.  
First, we can try to minimize
\begin{equation}
R_1 = \max_{f \in \mathcal{F}} \frac{\E [f(T)]}{f(T_f^*)}
\label{max_exp}
\end{equation}
which essentially gives a distribution over trees such that in expectation,
each function $f$ is well-approximated.  Second, and much more difficult, we
can look for an algorithm that uses the objective
\begin{equation}
R_2 = \E \left[\max_{f \in \mathcal{F}} \frac{f(T)}{f(T_f^*)}\right]
\label{exp_max}
\end{equation}
A bound on \eqref{exp_max} subsumes \eqref{max_exp} and proves there exists a
single tree that is simultaneously good for all $f$. 
We call $R_1$ the {\em oblivious} approximation ratio and $R_2$ the {\em simultaneous} approximation ratio.
In this paper, we will
work with the weaker, oblivious objective \eqref{max_exp}.

Both objectives have been studied in the literature.  The tree embeddings used by Awerbuch and Azar
\cite{awerbuch1997bbn} give an $O(\log^2 n)$ oblivious approximation,
which was later reduced to $O(\log n)$ \cite{bartal1998aam,
fakcharoenphol2003tba}.  Goel and Estrin \cite{goel2003soc} improved this to
$O(\log |\mathcal{D}|)$ and also prove the same bound on the stronger
simultaneous objective.  Gupta et al.\ \cite{gupta2006ond} achieve a
$O(\log^2 n)$ oblivious approximation for a generalization where
both the function and the demands are unknown.  Khuller et al.\
\cite{khuller1995bms} studied special case of simultaneously approximating
$f(x) = x$ and $f(x) = 1$ for $x \ge 1$, i.e.\ the shortest-path and Steiner
trees, and prove an $O(1)$ simultaneous approximation.  These 2 functions
constitute opposite extremes of functions in $\mathcal{F}$, and one may wonder
if an $O(1)$ approximation for these 2 functions also works for all $f \in
\mathcal{F}$ lying ``in-between''.  However, it is not difficult to construct
a graph and a set of demands such that the shortest-path and Steiner trees are
identical, but this tree is an $\omega(1)$-approximation for other $f \in
\mathcal{F}$.  Enachescu et al.\ \cite{enachescu2005sfa} achieve an $O(1)$
simultaneous value but only for grid graphs, assuming spatial
correlation among nearby nodes. This naturally leads to the following questions:

\begin{quote}
Is $R_1 = O(1)$ achievable? If yes, is there a polynomial algorithm that
guarantees $R_1 = O(1)$?
\end{quote}

We answer both questions in the affirmative. We first write a simple LP
formulation of the problem and show that using the ellipsoid method on the
dual we can find an $O(1)$ approximation to the optimal ratio, whatever it
happens to be for a given problem instance. We also show that given an
appropriate separation oracle the optimum is constant and compute an explicit
distribution over $1+\lceil\log \left(\sum_v d_v\right)\rceil$ trees in polynomial time. This general
approach is along the lines of small metric tree embeddings~\cite{ccggp:tree98} and
oblivious congestion minimization~\cite{h:congestion08}.

Our key result is the construction of the necessary separation oracle
subroutine, running in polynomial time, that proves a constant is achievable.
We build our oracle around the GMM algorithm for SSBaB, using a modified
analysis to solve a different problem in which we bound the cost of the GMM
tree by a combination of different trees under different cost functions. 
\subsection{Organization of the Paper}

In Section \ref{algorithm_section} we present an LP formulation and a
framework using an approximate separation oracle that finds a
constant-factor approximation to the optimal oblivious approximation ratio.
In Section \ref{oracle_section} we present our primary result, which proves
the oblivious approximation ratio is constant and constructs the separation
oracle required by Section \ref{algorithm_section} assuming some extra
conditions on the input, and in Section \ref{regularization_section} we
complete the proof by showing those extra assumptions can be removed. We
conclude with some open problems (including whether $R_2 = O(1)$ can be
achieved).

\section{LP Formulation and Algorithm Framework}
\label{algorithm_section}

Let $R_1$ be the worst-case optimal oblivious ratio, i.e.\
\[
R_1 = \max_{G,l,\mathcal{D},r} \min_{\mathcal{M}} \max_{f} \frac{\E_{T \sim \mathcal{M}}[f(T)]}{f(T_f^*)}
\]
where $\mathcal{M}$ is a distribution over $\mathcal{T}$.  In this section we discuss the problem of finding an $O(1)$-oblivious approximation if one exists.  

By losing a factor of $2$ in the approximation ratio we can restrict our analysis to a smaller class of aggregation functions.
Let $D = 2^{\lceil\log(\sum_v d_v)\rceil}$, the total amount of demand rounded up to the nearest power of 2.  We never route more than $D$ flow on any edge, and $d_v$ is integral, so we only care about $f(x)$ for integers $0 \le x \le D$.  Suppose $f \in \mathcal{F}$, and $2^i < x < 2^{i+1}$.  By the monotonicity of $f$, $f(2^i) \le f(x) \le f(2^{i+1})$, and by the concavity of $f$, $f(2^{i+1}) \le 2f(2^i)$, so with a loss of a factor of 2 we can interpolate between $f(2^i)$ and $f(2^{i+1})$ and assume $f$ is piecewise linear with breakpoints only at powers of 2.  Let $A_i(x) = \min\{x,2^i\}$ and $T_i^*$ the optimal aggregation tree for $A_i$.  We call $A_i(x)$ the $i$-th atomic function following the terminology of Goel and Estrin \cite{goel2003soc}, and it is easy to see that any $f \in \mathcal{F}$ that is linear between successive powers of 2 can be written as a linear combination of $\{ A_i \}_{0\le i \le \log D}$.  Therefore, it suffices to design an algorithm $\mathcal{A}$ minimizing $\max_{i} \E_\mcA[A_i(T_{\mcA})]/A_i(T_i^*)$.

Our algorithm makes use of the standard SSBaB problem where $f$ is known.  We assume that $f$ is given in the form of a set of $K$ pipes $\{(\sigma_k, \delta_k)\}_{0 \le k \le K-1}$, where the cost of routing $x$ flow on pipe $k$ is equal to $\sigma_k + x\delta_k$.  Then $f(x)$ is defined as the cost of using the cheapest pipe for $x$ flow: $\min_k \sigma_k + x\delta_k$.  We assume that $\sigma_0 \le \sigma_1 \le \cdots \le \sigma_{K-1}$, and by concavity we can assume $\delta_0 \ge \delta_1 \ge \cdots \ge \delta_{K-1}$.  Define $u_k = \frac{\sigma_k}{\delta_k}$, the point at which the cost due to $\delta_kx$ begins to outweigh the cost due to $\sigma_k$.   We call $u_k$ the \emph{capacity} of pipe $k$; the name arises from an alternate formulation
(equivalent up to a factor of 2)
of SSBaB where pipes have a fixed cost $\sigma_k$ for a fixed capacity $u_k$.  Let $\pi_{BaB}$ be the best-known approximation ratio for SSBaB.  Currently $\pi_{BaB} = 24.92$ using an algorithm by Grandoni and Italiano~\cite{grandoni2006ias}.

We also employ an approximation algorithm for a special case of SSBaB, the single-sink rent-or-buy (SSRoB) problem.  Here $f(x)$ is characterized by 2 pipes: $(0,1)$ and $(M,0)$, i.e.\ we can pay $x$ to route $x$ flow or pay $M$ to route any amount of flow.  Let $\pi_{RoB}$ be the best-known SSRoB approximation ratio.  Eisenbrand et al.~\cite{eisenbrand2008acf} give a $2.92$-approximation.

If we can calculate $A_i(T)$ and $A_i(T_i^*)$ for every $i$ and $T \in \mathcal{T}$ then the following linear program finds the optimal distribution of trees.
\begin{equation}\label{primalLP}
\begin{array}{rlll}
\min & \theta & &\\
\textrm{s.t.} & & \sum_{T \in \mathcal{T}} x_T &\ge 1 \\ 
&  \forall 0 \le i \le \log D, & \theta A_i(T_i^*) - \sum_{T \in \mathcal{T}}x_T A_i(T) &\ge 0 \\
& & x,\theta & \ge 0 \\
\end{array}
\end{equation}
In other words, we want a distribution $\{ x_T\}_{T \in \mathcal{T}}$ of trees minimizing $\max_i \frac{\sum_T x_tA_i(T)}{A_i(T_i^*)}$.   However, this approach is not directly tractable, as $T_i^*$ is $NP$-hard to find, and $|\mathcal{T}|$ is exponentially large.

We solve an SSRoB approximation for each $A_i$ to get $A_i(\tilde{T_i})$---a $\pi_{RoB}$-approximation---and replace $A_i(T_i^*)$ with $A_i(\tilde{T_i})$ in the constraints, so that all quantities in the LP are polynomial-time computable.  
Now consider the dual of \eqref{primalLP}, which is given by
\begin{equation}\label{dualLP}
\begin{array}{rlll}
\max & \beta & &\\
\textrm{s.t.} & &\sum_{i = 0}^{\log D} \alpha_iA_i(\tilde{T_i}) &\le 1 \\
&  \forall T \in \mathcal{T} & \beta - \sum_{i = 0}^{\log D}\alpha_iA_i(T) &\le 0 \\
& & \alpha,\beta & \ge 0 \\
\end{array}
\end{equation}
With an approximate separation oracle for the dual \eqref{dualLP}, we can approximate the solution in polynomial time using the ellipsoid method, and then transform it into an approximate solution to the primal \eqref{primalLP}.  More formally:

\begin{thm} \label{first_alg_thm} With a randomized $\pi_{BaB}$-approximation to SSBaB, we can find a $2\pi_{RoB}\pi_{BaB}R_1$-approximation in expectation to the primal LP \eqref{primalLP} that runs in polynomial time with high-probability.
\end{thm}

The proof uses a SSBaB approximation algorithm to construct an approximate separation oracle for \eqref{dualLP}.  However, we will not prove this theorem because it is a special case of the following more general result, assuming that $R_1$ is a constant which will follow from Theorem \ref{gmm_thm}.

\begin{thm}  \label{alg_thm}
If there exists a polynomial-time algorithm $\mcA$ and a given constant $c$ such that $\forall\,\alpha_0, \ldots, \alpha_{K-1} \ge 0$, $\mcA$ finds $T_{\mcA}$ such that $\E_{\mcA}\left[\sum_i \alpha_iA_i(T_{\mcA})\right] \le c\sum_i \alpha_i A_i(T_i^*)$ then we can construct an algorithm that runs in polynomial-time with high probability, makes $O(\textnormal{poly}(\log D))$ calls to $\mcA$ with high probability, and achieves an expected oblivious approximation ratio of $2c\pi_{RoB}$ using a distribution over $1+\log D$ trees.
\end{thm}

Proving that such an algorithm $\mcA$ exists for a constant $c$ is the primary result of this paper and is discussed in sections \ref{oracle_section} and \ref{regularization_section}.

\begin{rem} If $\mcA$ is deterministic then the algorithm always runs in polynomial time and the expected ratio is $c\pi_{RoB}$, and if it is randomized then the algorithm runs in polynomial time with high probability and the expected ratio is $2c\pi_{RoB}$.  For randomized $\mcA$ the ratio can also be reduced to $(1+\epsilon)c\pi_{RoB}$ with a $\frac{1}{\epsilon}$-factor increase in the runtime.
\end{rem}

\begin{proof}[Proof of Theorem \ref{alg_thm}]

Let $A_i(\tilde{T_i})$ be a $\pi_{RoB}$-approximation to $A_i(T_i^*)$ as above.  We construct an approximate separation oracle $\mathcal{S}(\valpha, \beta)$ for the dual \eqref{dualLP} as follows:

\begin{enumerate}

\item Check if $\sum_i \alpha_iA_i(\tilde{T_i}) > 1$.  If so, we have a violated constraint and are done.

\item \label{find_tree_step} Run $\mcA(\valpha)$ until it returns a tree $T$ such that $\sum_i \alpha_iA_i(T) < 2c\sum_i \alpha_iA_i(\tilde{T_i})$.

\item If $\sum_i \alpha_iA_i(T) < \beta$, return $T$.  Otherwise, return feasible.

\end{enumerate}

For a fixed $\beta$, let $\mathcal{P}_\beta$ be the polytope defined by $\sum_i \alpha_iA_i(\tilde{T_i}) \le 1$, and $\beta - \sum_i \alpha_iA_i(T) \le 0$ for all $T \in \mathcal{T}$.  We run the following procedure to find the desired distribution of trees:

\begin{enumerate}

\item Run the ellipsoid method to check the feasibility of $\mathcal{P}_{2c}$, starting with the initial bounding box $0 \le \alpha_i \le 1$ $\forall i$ and using $\mathcal{S}$ as the separation oracle.  It will terminate as infeasible.

\item Let $\mathcal{C}$ be the set of constraints returned by $\mathcal{S}$ proving $\mathcal{P}_{2c}$ is infeasible.  It consists of $\sum_{i = 0}^{\log D} \alpha_iA_i(\tilde{T_i}) \le 1$, and $2c - \sum_{i = 0}^{\log D}\alpha_iA_i(T) \le 0$ for $T$ in some subset of trees $\mathcal{T}'$.

\item In the dual LP (2), restrict the constraints to $\mathcal{C}$, and take the dual to get
\begin{equation}
\label{small_primal}
\begin{array}{rlll}
\min & \theta & &\\
\textrm{s.t.} & & \sum_{T \in \mathcal{T}'} x_T &\ge 1 \\ 
&  \forall 0\le i \le \log D, & \theta A_i(\tilde{T_i}) - \sum_{T \in \mathcal{T}'}x_T A_i(T) &\ge 0 \\
& & x,\theta & \ge 0 \\
\end{array}
\end{equation}

\item Find a vertex optimal solution to \eqref{small_primal}, and return the distribution $\{x^*_T\}$.
\end{enumerate}

First, we claim that $\mathcal{S}(\valpha, \beta)$ will find a violated constraint whenever $\beta \ge 2c$ and will do so in polynomial time with high probability.  If $\sum_i \alpha_iA_i(\tilde{T_i}) \le 1$  is violated, then we are done.  If not, we know $\mcA(\valpha)$ finds $T_\mcA$ such that 
\[
\E_{\mcA}\left[\sum_i \alpha_iA_i(T_{\mcA})\right] \le c\sum_i \alpha_i A_i(T_i^*) \le c\sum_i \alpha_iA_i(\tilde{T_i}) \le c
\]
By Markov's inequality $\Pr_{\mcA} \left[\sum_i \alpha_iA_i(T_{\mcA}) \ge 2c\sum_i \alpha_iA_i(\tilde{T_i})\right] \le \frac{1}{2}$, so with high probability $O(\log n)$ invocations of $\mcA$---each running in polynomial time---suffice in step \ref{find_tree_step} of $\mathcal{S}$ to find a $T$ satisfying $\sum_i \alpha_iA_i(T) < 2c\sum_i \alpha_iA_i(\tilde{T_i})$.  Now if $\beta \ge 2c$, the constraint $\beta - \sum_i \alpha_iA_i(T) \le 0$ is violated.

With the necessary separation oracle, the ellipsoid algorithm can solve feasibility of $\mathcal{P}_\beta$ in $O(\textnormal{poly}(\log D))$ iterations, so using $\mathcal{S}$ it will conclude $\mathcal{P}_{2c}$ is infeasible\footnote{In practice $\mcA$ may find violated constraints for $\beta < 2c$, and we can do binary search to find the smallest infeasible $\beta$.  However, we cannot improve the provable guarantee beyond $\beta = c$, and this comes at a cost to the runtime.}.  
The set of constraints $\mathcal{C}$ returned by $\mathcal{S}$ during the execution constitutes a proof of infeasibility,
and $\mathcal{C}$ consists of $\sum_{i = 0}^{\log D} \alpha_iA_i(\tilde{T_i}) \le 1$, and $\beta - \sum_{i = 0}^{\log D}\alpha_iA_i(T) \le 0$ for each $T$ in some set of trees $\mathcal{T}'$.  

Consider writing \eqref{dualLP} with only the constraints in $\mathcal{C}$.  Taking the dual yields \eqref{small_primal}, which only has variables $x_T$ for $T \in \mathcal{T}'$.  The ellipsoid algorithm concluded $\mathcal{P}_{2c}$ is infeasible after $O(\textnormal{poly}(\log D))$ iterations, so $|\mathcal{T}'|$ is only polynomially-large in the input size, implying we can solve \eqref{small_primal} exactly in polynomial time.

Find a vertex-optimal solution $\theta^*, x_T^*$ to \eqref{small_primal}.  The constraints in $\mathcal{C}$ are enough to restrict the optimal dual objective to be at most $2c$, so by duality $\theta^* \le 2c$.  Therefore, for all $i$
\[
\sum_{T \in \mathcal{T}'} x_T^*A_i(T) \le \theta^*A_i(\tilde{T_i}) \le 2cA_i(\tilde{T_i}) \le 2c\pi_{RoB}A_i(T_i^*)
\]
Divide by $A_i(T_i^*)$ to get the oblivious ratio:
\[
\max_i \frac{\sum_T x_T^*A_i(T)}{A_i(T_i^*)} \le 2c\pi_{RoB}
\] 

Moreover, we claim $\{x_T^*\}$ is a distribution over only $1+\log D$ trees.  The LP \eqref{small_primal} has $|\mathcal{T}'| + 1$ variables and $2 + \log D$ constraints, and the vertex-optimal solution $\theta^*, x_T^*$ must have $|\mathcal{T}'| + 1$ tight constraints, implying at least $|\mathcal{T}'| - \log D-1$ non-negativity constraints must be tight.  We know $\theta^*$ is positive, so only at most $1+\log D$ of the variables $x_T$ can be non-zero.
\end{proof}

\section{The Separation Oracle Subroutine $\mcA$}
\label{oracle_section}

By Theorem \ref{first_alg_thm} we can find an $O(1)$-approximation to $R_1$, whatever it may be, but it remains to prove that this optimal ratio is a constant.  In this section we construct the procedure $\mcA$ required by Theorem \ref{alg_thm} using the GMM algorithm for SSBaB.

Our contribution is adapting a special case of the analysis of the GMM algorithm, namely those cases that arise when $f(x) = \sum_i \alpha_iA_i(x)$, to solve a different problem--that of bounding the cost of the output by $\sum_i \alpha_i A_i(T_i^*)$ rather than $f(T_f^*)$.  The GMM algorithm and proof works in stages and bounds the cost of the pipes laid in each stage by a different chunk of the optimal tree $T_f^*$.  On the other hand, in our proof we bound the cost of each stage by the cost of a \emph{different} tree evaluated under a \emph{different} cost function.

\subsection{Background: The GMM Algorithm}

For completeness, we summarize the GMM algorithm and the key lemmas and definitions.  See the original paper \cite{guha2001cfa} for a thorough treatment.  We are given a graph, demands $\mathcal{D}$, and pipes $\{(\sigma_k, \delta_k)\}_{k \in [K]}$ as described in Section \ref{algorithm_section}.  
We assume the costs of successive pipes differ ``significantly'': for some constant $\gamma$ such that $0 < \gamma < \frac{1}{2}$, we have that $\delta_{k+1} < \gamma\delta_k$ and $\sigma_k < \gamma\sigma_{k+1}$.  For the SSBaB problem, it is easy to satisfy these constraints for arbitrary pipes with only an $O(1)$-factor loss.  For our problem, it is harder but still possible, and this is discussed in Section \ref{regularization_section}.  

We define $g_k$ as the indifference point between pipe $k$ and $k+1$, which is the solution to the equation $\sigma_k +\delta_kg_k = \sigma_{k+1} + \delta_{k+1}g_k$, and we define $b_k$ as the solution to $\sigma_{k+1} + \delta_{k+1}b_k = 2\gamma(\sigma_k + \delta_kb_k)$, which we interpret as the point at which pipe $k+1$ becomes ``significantly'' cheaper than pipe $k$.  It is easy to see that $u_k \le b_k \le u_{k+1}$ for all $k$.

The algorithm uses $O(1)$-approximations for Steiner tree and load-balanced facility location (LBFL), a generalization of the standard facility location problem.  In the LBFL problem we have a graph and demands as in SSBaB, a facility cost $F_v$ for each node $v$, and a lower bound $L_v$ on the demand that a facility at $v$ must service.  The objective is to choose facilities and routing paths so as to minimize the sum of the cost of the open facilities and the distances traveled by the demands to a servicing facility.  To approximate the LBFL we must relax the lower bound.  Using \cite{guha2000hpa} we can approximate the optimal LBFL cost to within $2\pi_F$ while reducing the lower bound by a factor of at most 3.  Here $\pi_F$ denotes the best approximation to the normal facility location problem, currently $\pi_F = 1.52$ by Mahdian et al.\ \cite{mahdian2002iaa}.  We use $\pi_{S}$ to denote the best approximation ratio for Steiner tree, currently $1.55$ due to Robins and Zelikovsky \cite{robins2000ist}.

Now we can describe the GMM algorithm itself.  At stage $k$, we lay pipe type $k$, and we break each stage into a Steiner tree step and a ``shortest-path'' tree step based on whether the cost of pipe $k$ is dominated by the term $\sigma_k$ or the term $\delta_kx$.  The effective demands will also change each stage.  Let $\mathcal{D}^{(k)}$ be the demand nodes at the start of stage $k$, and $d_v^{(k)}$ the stage $k$ demand at $v \in \mathcal{D}^{(k)}$.  Initially $\mathcal{D}^{(0)} = \mathcal{D}$.

\begin{description}

\item{1. \emph{Steiner Tree: }}
Find a $\pi_S$-approximate Steiner tree on $\mathcal{D}^{(k)} \cup \{r\}$ with edge cost per unit length $\sigma_k$.  Route all demands toward $r$.  Cut the farthest-upstream edge with more than $u_k$ flow, recalculate the flow, and repeat to get a forest with at least $u_k$ flow at each root other than $r$ and at most $u_k$ flow on each edge.

\item{2. \emph{Consolidation:}}
Let $t$ be a subtree not containing $r$ and $S_t$ the demand nodes in $\mathcal{D}^{(k)}$ it contains.  Choose $v \in S_t$ with probability $\frac{d_v^{(k)}}{\sum_{u \in S_t} d_u^{(k)}}$ and route all demand in $t$ back to $v$ using pipe $k$.

\item{3. \emph{Shortest Path Tree:}}
Approximately solve a LBFL problem with facility lower bound $b_k$ and edge cost per unit length $\delta_k$ on the \emph{original} demands $\mathcal{D}$ (not $\mathcal{D}^{(k)}$ and $d_v^{(k)}$).  This creates a forest of shortest-path trees with at least $b_k$ flow at each root.  If $b_k$ demand does not exist, route everything to $r$.

\item{4. \emph{Consolidation:}}
Let $t$ be subtree in the above forest servicing the demands $S_t$ in $\mathcal{D}$.  Choose $v \in S_t$ with probability $\frac{d_v}{\sum_{u \in S_t}  d_u}$, and route the true, current demand $d_v^{(k)}$ in $S_t$ back to $v$.  Let $\mathcal{D}^{(k+1)}$ be the set of nodes chosen for consolidation and $d_v^{(k+1)}$ the demand at these nodes after consolidation.
\end{description}

Next, we mention the crucial lemmas in the GMM analysis used in our proof.  See \cite{guha2001cfa} for the proofs.

\begin{lem}[GMM Lemma 4.1] \label{consolidation_lem} 
Let $\hat{d_v}$ be the current demand at some $v \in \mathcal{D}$ immediately after any consolidation step.  Then $\E[\hat{d_v}] = d_v$, i.e.\ the original demand.
\end{lem}

Using an algorithm that is a 3-approximation to the LBFL facility lower bounds, we have the following:

\begin{lem}[GMM Lemma 4.5] \label{current_demand_lem}
For every $v \in \mathcal{D}^{(k)}$, we have $\E[d_v^{(k)}] \ge \frac{b_{k-1}}{3}$.
\end{lem}

Define $P_k^{\delta}$ to be the \emph{incremental} cost (due to $\delta$) of the pipes laid in the \emph{facility location} step in stage $k$ and $P_k^{\sigma}$ to be the \emph{fixed} cost (due to $\sigma$) of the pipes laid in the \emph{Steiner tree} step in stage $k$.  All of the other costs incurred by the GMM algorithm can be bounded by $P_k^{\delta}$ and $P_k^{\sigma}$, so our analysis need only consider these quantities:

\begin{lem}[GMM Lemmas 4.2, 4.4, and 4.8] \label{gmm_cost_lem}  Let $P_k^{\delta}$ and $P_k^{\sigma}$ as defined above.
Then $\E[f(T_{GMM})] \le 4\sum_k \E[P_k^{\delta} + P_k^{\sigma}]$, where $T_{GMM}$ is the final tree.
\end{lem}

\subsection{Adapting the GMM Algorithm}

From Theorem \ref{alg_thm} we are given $\valpha$ such that $\alpha_i \ge 0$, and $\sum_i \alpha_i A_i(\tilde{T_i}) \le 1$.  We want to find a tree $T$ using the GMM algorithm such that $\sum_i \alpha_iA_i(T) \le c\sum_i \alpha_iA_i(T_i^*)$. 
Define $L =  \sum_i \alpha_iA_i(T_i^*)$, the \emph{multi-level cost}, and $f(x) = \sum_i \alpha_iA_i(x)$, the concave cost function.  Using this notation our objective becomes to find $T$ such that $f(T) \le cL$.  Define $K$ as the number of non-zero $\alpha_i$, and for $0 \le k \le K-1$ define $p(k) = j$  where $j$ is the index of the $k$-th non-zero $\alpha_i$. 

First, we claim that given $\valpha$ we can define the pipes $\{(\sigma_k, \delta_k)\}$ used by the GMM algorithm, and given SSBaB pipes satisfying some minor conditions we can recover $\valpha$.  The following lemmas characterize the equivalence between the 2 types of parameters:

\begin{lem} \label{alphatodelta_lem}
Given $\valpha$ satisfying $\alpha_i \ge 0$ with $K$ non-zero $\alpha_i$, the SSBaB pipes $\{(\sigma_k, \delta_k)\}_{0 \le k \le K}$ defined by $\delta_k = \sum_{j \ge k} \alpha_{p(j)}$ and $\sigma_k = \sum_{j < k} \alpha_{p(j)}2^{p(j)}$ define the function $f(x)$.  That is, $f(x) = \sum_i \alpha_iA_i(x) = \min_k \{ \sigma_k + \delta_k x\}$.
\end{lem}

\begin{lem} \label{deltatoalpha_lem}
Suppose we are given $K+1$ SSBaB pipes $\{(\sigma_k, \delta_k)\}_{0 \le k \le K}$ such that $\sigma_0 = 0$ and $g_k$ is a power of 2 for all $k$. For $0 \le k \le K-1$, let $p(k) = \log g_k$, $\alpha_{p(k)} = \delta_k - \delta_{k+1}$, and $\alpha_j =0$ whenever $j \neq p(k)$ for all $k$.  Then $\sum_i \alpha_iA_i(x) = \min_k \{ \sigma_k + \delta_k x\}$.
\end{lem}

\begin{proof}[Proof of Lemma \ref{alphatodelta_lem}]
By definition $f(x) = \sum_k \alpha_{p(k)}A_{p(k)}(x)$.  For any $k$, $f(x)$ is linear from $2^{p(k-1)}$ to $2^{p(k)}$ (we will assume $2^{p(-1)} = 0$ for consistency of notation), which will correspond to pipe $k$.  For $x \in [2^{p(k-1)},2^{p(k)}]$, the functions $A_{p(0)}(x), \ldots, A_{p(k-1)}(x)$ have leveled off, and $A_{p(k)}(x), \ldots, A_{p(K-1)}(x)$ are growing at rate 1.  Define $\delta_k$ as the slope of $f(x)$ in this interval: $\delta_k = \sum_{j \ge k} \alpha_{p(j)}$.

Now we can define $\sigma_k$ to match $f(x)$ in the interval $[2^{p(k-1)},2^{p(k)}]$:
\begin{align*}
\sigma_k + \delta_k2^{p(k-1)} = \sum_i \alpha_iA_i(2^{p(k-1)}) 
&= \sum_{j < k} \alpha_{p(j)}2^{p(j)} + \sum_{j \ge k} \alpha_{p(j)}2^{p(k-1)} \\
&= \sum_{j < k} \alpha_{p(j)}2^{p(j)} + \delta_k2^{p(k-1)} \\
\Rightarrow \sigma_k &= \sum_{j < k} \alpha_{p(j)}2^{p(j)}
\end{align*}
We also add a $K+1$st pipe such that $\delta_K = 0$ and $\sigma_K = \sum_k \alpha_{p(k)}2^{p(k)}$ to cover the interval after every $A_{p(k)}$ has leveled off.
Now, we claim $f(x) = \min_j \{\sigma_j + \delta_jx\}$:  for each $k$ we know $f(x) = \sigma_k + \delta_kx$ whenever $x \in [2^{p(k-1)},2^{p(k)}]$ by our choice of $\delta_k$ and $\sigma_k$, and by the concavity of $f(x)$ for each $j$ we have $\sigma_j + \delta_jx > f(x)$ when $x < 2^{p(j-1)}$ or $x > 2^{p(j)}$.  Therefore no other pipe can be cheaper in this interval.
Concavity also ensures that $\sigma_k < \sigma_{k+1}$ and $\delta_k > \delta_{k+1}$ for all $k$, yielding valid SSBaB pipes.
\end{proof}

\begin{proof}[Proof of Lemma \ref{deltatoalpha_lem}]
Let $K+1$ be the number of pipes, and $\delta_0 > \cdots > \delta_{K}$, $0 = \sigma_0 < \cdots < \sigma_{K}$.  Since we never route more than $D$ flow we may assume the cost function levels off at some $x \le D$, so that $\delta_K = 0$.  Define $p(k) = \log g_k$ for $0 \le k \le K-1$: when we change pipes at $g_k$ the slope of $f(x)$ drops, which can occur only because the term $\alpha_{p(k)}A_{p(k)}(x)$ levels off.  Recover $\alpha_{p(k)}$ by reversing the definitions in the proof of Lemma \ref{alphatodelta_lem}:  we have $\delta_k = \sum_{j \ge k} \alpha_{p(j)}$, so for $k \le K-1$ let $\alpha_{p(k)} = \delta_k - \delta_{k+1}$.

We now show by induction that $\sum_k \alpha_{p(k)}A_{p(k)}(x) = \min_j \{\sigma_j + \delta_jx\}$.  For the base case $x \in [0,g_0]$, we have
\begin{align*}
\min_j \{\sigma_j + \delta_jx\} = \delta_0x = (\delta_0 - \delta_K)x =\sum_{k=0}^{K-1}(\delta_k - \delta_{k+1})x = \sum_k \alpha_{p(k)}x = \sum_k \alpha_{p(k)}A_{p(k)}(x)
\end{align*}
Now assume that for $x \in [0,g_{i-1}]$ that $\sum_k \alpha_{p(k)}A_{p(k)}(x) = \min_j \{\sigma_j + \delta_jx\}$.  For $x \in (g_{i-1}, g_i]$, we know that $f(x) = \sigma_i + \delta_ix$.  Therefore,
\begin{align*}
\sigma_i + \delta_ix =& \left(\sigma_{i-1}+\delta_{i-1}2^{p(i-1)}\right) + \delta_i(x - 2^{p(i-1)}) \\
=&\sum_k \alpha_{p(k)}A_{p(k)}(2^{p(i-1)}) + \sum_{k=i}^{K-1}(\delta_k - \delta_{k+1})(x - 2^{p(i-1)}) \\
=& \sum_{k < i} \alpha_{p(k)}A_{p(k)}(2^{p(i-1)}) + \sum_{k \ge i} \alpha_{p(k)}x \\
=& \sum_k \alpha_{p(k)}A_{p(k)}(x)
\end{align*}
We use that pipes $i-1$ and $i$ have equal cost at $g_{i-1}$ in the first line and the induction hypothesis in the second line.
\end{proof}

We note that $\alpha_{p(k)}$ corresponds not to a particular SSBaB pipe, but to a breakpoint between pipes:  when we switch from pipe $k$ to $k+1$ at $2^{p(k)}$ flow, the slope of $f$ drops from $\delta_k$ to $\delta_{k+1}$, which is caused by the term $\alpha_{p(k)}A_{p(k)}(x)$ leveling off.

Given the above equivalence, we will use $\valpha$ and $\{(\sigma_k, \delta_k)\}_k$ interchangeably for the remainder of the paper, using whichever representation is more convenient and converting from one form to another using Lemmas \ref{alphatodelta_lem} and \ref{deltatoalpha_lem}.  However, the additional constraints that for some parameter $0 < \gamma < \frac{1}{2}$ we have $\delta_{k+1} < \gamma\delta_k$ and $\sigma_k < \gamma\sigma_{k+1}$ for all pipes $k$, will restrict the possible vectors $\valpha$ that can be run through the algorithm:

\begin{defn} Call $\valpha$ \emph{$\gamma$-regular} if the pipes found using Lemma \ref{alphatodelta_lem} satisfy $\delta_{k+1} < \gamma\delta_k$ and $\sigma_k < \gamma\delta_{k+1}$. 
\end{defn}

We note the following constraints that $\gamma$-regularity imposes on $\valpha$:

\begin{lem} \label{alpha_delta_lem}
If $\delta_{k+1} < \gamma\delta_k$, then $\alpha_{p(k)} > (1-\gamma)\delta_k$ and $\alpha_{p(k)} > \frac{1-\gamma}{\gamma}\alpha_{p(k+1)}$.
\end{lem}

\begin{proof} 
These follow immediately from $\alpha_{p(k)} = \delta_k - \delta_{k+1}$ and $\delta_{k+1} < \gamma\delta_k$.
\end{proof}

\subsection{Approximation guarantee assuming regular $\valpha$}

We will first prove the existence of the separation oracle procedure $\mcA$ in Theorem \ref{alg_thm} for $\gamma$-regular $\valpha$ and later prove in Section \ref{regularization_section} that arbitrary $\valpha$ can be regularized with only an $O(1)$ change in $f(x)$ and $L$:

\begin{thm} \label{gmm_thm}
Let $\valpha$ be $\gamma$-regular, and let $f(x) = \sum_i \alpha_iA_i(x)$, and $L = \sum_i \alpha_iA_i(T_i^*)$.  Then the GMM algorithm finds a tree $T_{GMM}$ such that $\E\left[f(T_{GMM})\right] = O(L)$.
\end{thm}

Roughly, our proof bounds the cost of the pipes laid in phase $k$ of the algorithm by $\alpha_{p(k)}A_{p(k)}(T_{p(k)}^*)$.  Using Lemma \ref{gmm_cost_lem} we concentrate on $P_k^{\delta}$ and $P_k^{\sigma}$ and ignore the other costs.  
First, we bound the cost of the Steiner tree steps:

\begin{lem} 
Let $\pi_S$ be the approximation ratio for Steiner tree.  Then we have $\sum_k \E[P_k^{\sigma}] \le \frac{3\pi_S}{1-\gamma}L$.\label{steiner_cost_lem}
\end{lem}
\begin{proof}
We need to bound the cost of a Steiner tree spanning the current demands $\mathcal{D}^{(k)}$ with cost per unit length $\sigma_k$.  If $k=0$, then $\sigma_k = 0$ and we have nothing to bound, so assume $k > 0$.

We use the edges in $T_{p(k-1)}^*$.  Note that it spans $\mathcal{D} \cup \{r\}$ and hence $\mathcal{D}^{(k)} \cup \{r\}$, and let $W_k \subseteq T_{p(k-1)}^*$ be the subset of edges spanning these nodes.  
By Lemma \ref{current_demand_lem} each $v \in \mathcal{D}^{(k)} $ has aggregated at least $\E[d_v^{(k)}] \ge \frac{b_{k-1}}{3}$ demand.  At the end of the previous LBFL phase, we chose a node $v$ for consolidation from the set of all $u$ routing to facility $f$ with probability $\frac{d_v}{\sum_{u\rightarrow f} d_u} \le \frac{3d_v}{b_{k-1}}$.  An edge is in $W_k$ only if some $v \in \mathcal{D}^{(k)}$ routes through it, so by the union bound an edge carrying $x_e^*$ demand in $T_{p(k-1)}^*$ is in $W_k$ with probability at most $\frac{3x_e^*}{b_{k-1}}$.

The tree $W_k$ pays $\sigma_k$ for any amount of flow, whereas $T_{p(k-1)}^*$ pays $A_{p(k-1)}(x_e^*) = \min\{2^{p(k-1)}, x_e^*\}$ to send $x_e^*$ flow on $e$.  Then the cost of $W_k$ is
\begin{equation}  \label{steiner_cost_eq}
\begin{split} 
\E[W_k] &= \sigma_k\sum_e \Pr[e \in W_k]l_e
= \sigma_k \sum_e \Pr[e \in W_k]l_e\frac{A_{p(k-1)}(x_e^*)}{\min\{x_e^*,2^{p(k-1)}\}} \\
&\le \sigma_k \sum_{e : x_e^* \le 2^{p(k-1)}} \frac{3x_e^*}{b_{k-1}} \frac{A_{p(k-1)}(x_e^*)}{x_e^*}l_e 
+ \sigma_k \sum_{e : x_e^* > 2^{p(k-1)}} 1\cdot\frac{A_{p(k-1)}(x_e^*)}{2^{p(k-1)}}l_e \\
&= 3\frac{\sigma_k}{b_{k-1}} \sum_{e : x_e^* \le 2^{p(k-1)}} A_{p(k-1)}(x_e^*)l_e 
+ \frac{\sigma_k}{2^{p(k-1)}} \sum_{e : x_e^* > 2^{p(k-1)}} A_{p(k-1)}(x_e^*)l_e 
\end{split}
\end{equation}

We need to bound $\frac{\sigma_k}{b_{k-1}}$ and $\frac{\sigma_k}{2^{p(k-1)}}$.  For the former term, 
\[
\frac{\sigma_k}{b_{k-1}} = \frac{\sigma_k(2\gamma\delta_{k-1} - \delta_k)}{\sigma_k - 2\gamma\sigma_{k-1}} 
\le \frac{\sigma_k(2\gamma\delta_{k-1} - \delta_k)}{2\gamma\sigma_k(1- \gamma)} 
\le \frac{2\gamma(\delta_{k-1} - \delta_k)}{2\gamma(1- \gamma)} 
= \frac{\alpha_{p(k-1)}}{1-\gamma}
\]
using that $b_{k-1} = \frac{\sigma_k - 2\gamma\sigma_{k-1}}{2\gamma\delta_{k-1} - \delta_k}$ by definition, the $\gamma$-regularity constraints on $\sigma_{k-1}$, and the fact that $2\gamma < 1$.
For the latter term,
\begin{align*}
\frac{\sigma_k}{2^{p(k-1)}} = \frac{\sigma_{k-1} + \alpha_{p(k-1)}2^{p(k-1)}}{2^{p(k-1)}} &\le \frac{\gamma\sigma_k + \alpha_{p(k-1)}2^{p(k-1)}}{2^{p(k-1)}} = \gamma\frac{\sigma_k}{2^{p(k-1)}} + \alpha_{p(k-1)} \\
\Rightarrow (1-\gamma)\frac{\sigma_k}{2^{p(k-1)}} &\le \alpha_{p(k-1)} \Rightarrow \frac{\sigma_k}{2^{p(k-1)}} \le \frac{\alpha_{p(k-1)}}{1-\gamma}
\end{align*}
using the formula for $\sigma_k$ in Lemma \ref{alphatodelta_lem} and $\gamma$-regularity.

Plug these into the final line in equation \eqref{steiner_cost_eq} above:
\begin{align*}
\E[W_k] \le& \frac{\alpha_{p(k-1)}}{1-\gamma} \left(3\sum_{e : x_e^* \le 2^{p(k-1)}} A_{p(k-1)}(x_e^*)l_e +\sum_{e : x_e^* > 2^{p(k-1)}} A_{p(k-1)}(x_e^*)l_e\right) \\
= & \left(\frac{3}{1-\gamma}\right)\alpha_{p(k-1)}A_{p(k-1)}(T_{p(k-1)}^*)
\end{align*}

We lose another factor of $\pi_S$ in approximating the Steiner tree.  Sum over all $k$ to bound $\sum_k \E[P_k^{\sigma}]$ by $\frac{3\pi_S}{1-\gamma}L$.
\end{proof}

Analyzing the LBFL step requires an additional lemma bounding the difference between $g_k$ and $b_k$:

\begin{lem}  For every $k$, $g_k \le b_k \le \frac{1-2\gamma^2}{\gamma}g_k$.
\label{b_g_lem}
\end{lem}

\begin{proof}
The bound $g_k \le b_k$ follows from Lemma 3.5 in GMM \cite{guha2001cfa}.  For the other inequality, from the definition of $b_k$ and $g_k$ we have
\begin{align*}
g_k &= \frac{\sigma_{k+1}-\sigma_{k}}{\delta_{k}-\delta_{k+1}}  & b_k &= \frac{\sigma_{k+1}-2\gamma\sigma_k}{2\gamma\delta_k - \delta_{k+1}} 
&\Rightarrow \frac{b_k}{g_k} &= \frac{\sigma_{k+1}-2\gamma\sigma_k}{\sigma_{k+1}-\sigma_{k}} \cdot\frac{\delta_{k}-\delta_{k+1}}{2\gamma\delta_k - \delta_{k+1}} 
\end{align*}
For the ratio of $\sigma$ terms,
\begin{align*}
\frac{\sigma_{k+1} - 2\gamma\sigma_k}{\sigma_{k+1} - \sigma_k} &= \frac{\sigma_{k+1} - \sigma_k}{\sigma_{k+1} - \sigma_k} + (1-2\gamma)\frac{\sigma_k}{\sigma_{k+1} - \sigma_k} \\
& < 1 + (1-2\gamma)\frac{\sigma_k}{\left(\frac{1}{\gamma}-1\right)\sigma_k}  = 1 + \frac{\gamma-2\gamma^2}{1-\gamma}
= \frac{1-2\gamma^2}{1-\gamma}
\end{align*}
Similarly, for the $\delta$s,
\begin{align*}
\frac{\delta_k - \delta_{k+1}}{2\gamma\delta_k-\delta_{k+1}} 
&= \frac{2\gamma\delta_k - \delta_{k+1}}{2\gamma\delta_k - \delta_{k+1}} + (1-2\gamma)\frac{\delta_k}{2\gamma\delta_k - \delta_{k+1}} \\
&< 1 + (1-2\gamma)\frac{\delta_k}{(2\gamma-\gamma)\delta_k} = \frac{1-\gamma}{\gamma}
\end{align*}
Combining the 2 bounds,
\[
\frac{b_k}{g_k} \le \frac{1-2\gamma^2}{1-\gamma}\frac{1-\gamma}{\gamma} = \frac{1-2\gamma^2}{\gamma}
\]
\end{proof}

Now we can bound the LBFL cost $\E[P_k^{\delta}]$:

\begin{lem} We have that $\sum_k \E[P_k^{\delta}] \le 2\pi_F\frac{1-2\gamma^2}{\gamma-\gamma^2}L$ where $\pi_F$ is the approximation ratio for the standard (non-load-balanced) facility location problem.
\label{facility_cost_lem}
\end{lem}

\begin{proof}
In the shortest path tree step, the GMM algorithm solves an LBFL problem on the original demands $\mathcal{D}$ with facility lower bound $b_k$ and edge cost per unit length $\delta_k$.  
We will construct a feasible solution using the edges of $T_{p(k)}^*$.
Orient the edges towards $r$, and find the farthest upstream (i.e.\ away from $r$) edge routing at least  $b_{k}$ flow.  Cut the edge, and place a facility at the upstream node.  Subtract this flow from downstream edges, and repeat the procedure.  
If we finish with less than $b_k$ flow at the root node, we route each demand still reaching the root from its source vertex along the tree to the nearest existing facility (according to distances in $T_{p(k)}^*$).  Let $F_k$ be the resulting forest, and note that it has at least $b_{k}$ flow at each facility.

For an edge $e$ let $x_e$ be the amount $F_k$ routes on $e$ when the demands $\mathcal{D}$ are routed, and $x_e^*$ the amount that $T_{p(k)}^*$ routes on $e$.  We now show that $x_e \le x_e^*$.  If we finish cutting $T_{p(k)}^*$ with at least $b_k$ at the root then all flows are a subset of the flows in $T_{p(k)}^*$ so $x_e \le x_e^*$.  If we end up with too little demand for a facility in the final step then some of those demands will not be flowing downstream towards $r$ in $F_k$.  For each edge they take towards $r$, they are following the routing in $T_{p(k)}^*$, so $x_e \le x_e^*$.  For each $e$ edge taken away from $r$, we are no longer following $T_{p(k)}^*$, but we must be moving upstream towards the nearest facility.  This implies that in the tree $T_{p(k)}^*$ edge $e$ carried more than $b_k$ flow because all demand at the upstream facility flowed through $e$ towards $r$.  Since we are sending strictly less than $b_k$ demand upstream we still have $x_e \le x_e^*$.

The forest $F_k$ never routes more than $b_k$ flow, so $x_e \le b_k$.  
When $x_e^* \le g_k$, $x_e^* = A_{p(k)}(x_e^*)$, so $x_e \le A_{p(k)}(x_e^*)$.  
Since $A_{p(k)}$ levels off at $g_k$, this may not hold for $x_e^* > g_k$ , but by Lemma \ref{b_g_lem} $b_k \le \frac{1-2\gamma^2}{\gamma}g_k$.
Therefore $x_e \le b_k \le \frac{1-2\gamma^2}{\gamma}A_{p(k)}(x_e^*)$ when $x_e^* \ge g_k$.

Now let $y_e$ be the flow $F_k$ routes on edge $e$ when the current, stage $k$ demands $\mathcal{D}^{(k)}$ are used.  By Lemma \ref{consolidation_lem}, $\E[\hat{d_v}] = d_v$ for each $v \in \mathcal{D}$.  Summing over all the demands that contribute to an edge's flow, we have $\E[y_e] = x_e$.

The cost of $F_k$ with $\delta_j$ cost per unit edge length is
\begin{align*}
\E\left[\delta_k \sum_e l_ey_e\right] = \delta_k\sum_e l_e x_e \le \delta_k\sum_e l_e\left(\frac{1-2\gamma^2}{\gamma}A_{p(k)}(x_e^*)\right)
\le \left(\frac{\alpha_{p(k)}}{1-\gamma}\right)\left(\frac{1-2\gamma^2}{\gamma}\right)A_{p(k)}(T_{p(k)}^*)
\end{align*}
using $\frac{1-2\gamma^2}{\gamma} > 1$ and $\alpha_{p(k)} \ge (1-\gamma)\delta_k$ from Lemma \ref{alpha_delta_lem}.

We can find an approximate LBFL solution that is a $2\pi_F$-approximation to the optimal cost and reduces the facility lower bound by a factor of at most $3$.  Therefore
\[
\E[P_k^{\delta}] \le 2\pi_F\E[F_k] \le \left(2\pi_F\frac{1-2\gamma^2}{\gamma - \gamma^2}\right)\alpha_{p(k)}A_{p(k)}(T_{p(k)}^*)
\]

Sum over all values of $k$ to bound the expected cost by $2\pi_F\frac{1-2\gamma^2}{\gamma-\gamma^2}L$.
\end{proof}

\begin{proof}[Proof of Theorem \ref{gmm_thm}]
Combining the bounds in Lemmas \ref{gmm_cost_lem}, \ref{facility_cost_lem}, and \ref{steiner_cost_lem}:
\[
\E[f(T_{GMM})] \le 4\left(2\pi_F\frac{1-2\gamma^2}{\gamma-\gamma^2} + \frac{3\pi_S}{1-\gamma}\right)L
\]
\end{proof}

This completes the analysis of $\mcA$ for $\gamma$-regular $\valpha$.  If arbitrary $\valpha$ can be $\gamma$-regularized for some $0 < \gamma < \frac{1}{2}$ it follows that $R = O(1)$.

Recent algorithms for SSBaB are based on the 
Gupta, Kumar, and Roughgarden (GKR) algorithm \cite{gupta2003sab, gupta2007approximation}, which achieves a better approximation ratio than GMM with a simpler analysis, and one may wonder whether we could reap the same benefits by basing our proof around this algorithm instead.  One round of GKR is roughly equivalent to one round of GMM---starting with about $g_{k-1}$ demand at a subset of nodes and ending with about $g_k$ demand at a smaller subset---but the GKR analysis bounds the entire cost of a round using only one tree, whereas GMM requires two.  However, each tree required by GMM can be easily constructed from some $T_{i}^*$ in $O(\alpha_iA_i(T_i^*))$, but building the tree needed by GKR and within the right bounds seems trickier.  Note that Lemmas \ref{steiner_cost_lem} and \ref{facility_cost_lem}
use two different trees, $T_{p(k-1)}^*$ and $T_{p(k)}^*$, analyzed in two different ways, either fixed or linear cost per edge.  Although this conveniently matches the GMM algorithm, it also required for the proof to work.  
Using only a single Steiner tree on a subset of the nodes as in GKR allows less flexibility,
so a proof may require a different approach or more substantial changes to the original GKR analysis.

\section {Handling Arbitrary $\valpha$}
\label{regularization_section}

Given any $\valpha$, where $\alpha_i \ge 0$, defining $f(x)$, a concave cost function, and $L$, the multi-level cost, we need to find regular $\valpha'$ defining $f'(x)$ and $L'$ such that $f(x) = O(f'(x))$ $\forall x$, and $L' = O(L)$.  Then applying Theorem \ref{gmm_thm} to $\valpha'$ gives $f'(T_{GMM}) = O(L')$, and 
\[
f(T_{GMM}) = O(f'(T_{GMM})) = O(L') = O(L)
\]
satisfying the precondition of Theorem \ref{gmm_thm}.
Note that we can allow $f$ to grow and $L$ to shrink arbitrarily in the transformation to $f'$ and $L'$, but we need to bound increases in $L$ and decreases in $f$.  By scaling by $\sum_i \alpha_i$ we may assume without loss of generality that $\sum_i \alpha_i = 1$.

First, we prove a simple bound on the change between each term $A_i(T_i^*)$ in $L$.

\begin{lem}  For any $i$ and any $k > 0$, $A_i(T_i^*) \le A_{i+k}(T_{i+k}^*) \le 2^kA_i(T_i^*)$.
\label{L_term_lem}
\end{lem}

\begin{proof}
Note $A_i(x) \le A_{i+k}(x) \le 2^kA_i(x)$ for $k > 0$.  Therefore 
\[
A_i(T_i^*) \le A_i(T_{i+k}^*) \le A_{i+k}(T_{i+k}^*) \le A_{i+k}(T_i^*) \le 2^kA_i(T_i^*)
\]
\end{proof}

To regularize the values we run $\valpha$ through a series of three procedures, one for each of the following lemmas, 
each of which changes $\valpha$ to satisfy an additional set of constraints.  None of the procedures are conceptually difficult, but the details are quite intricate.  We will state the lemmas, give a brief sketch of the ideas, and present the complete proofs in the appendix.

The first lemma is only a helper used in satisfying the $\sigma$ constraints.  The proof serves as a warmup for the later lemmas, which use similar ideas but are more involved.

\begin{lem} 
\label{max_capacity_lem}
Given arbitrary $\valpha$, we can find $\valpha'$ such that the corresponding $f'$,$L'$, $\delta'$, $\sigma'$ satisfy $f(x) \le f'(x)$, $L' \le 2L$, and $\frac{\sigma'_{K-1}}{\delta'_{K-1}} \le D$, where $K$ is the number of pipes, and $D$ is the total demand rounded up to a power of 2.
\end{lem}

The following 2 lemmas perform the actual regularization.

\begin{lem} 
\label{delta_lem}
Given $\valpha$ satisfying $\frac{\sigma_{K-1}}{\delta_{K-1}} \le D$, we can find $\valpha'$ 
such that the corresponding $f'$,$L'$, $\delta'$, $\sigma'$ satisfy
$f(x) \le 3f'(x)$, $L' = O(L)$, $\frac{\sigma'_{K-1}}{\delta'_{K-1}} \le D$, and $\delta'_{k+1} < \gamma\delta'_{k}$ for all $k$.
\end{lem}

\begin{lem} 
\label{sigmalem}
Given $\valpha$ satisfying $\frac{\sigma_{K-1}}{\delta_{K-1}} \le D$ and $\delta_{k+1} < \gamma\delta_{k}$,
we can find $\valpha'$ such that such the corresponding $f'$,$L'$, $\delta'$, $\sigma'$ satisfy
$f(x) \le \frac{5}{2}f'(x)$, $L'  = O(L)$, $\delta'_{k+1} < \gamma\delta'_{k}$, and $\sigma_k' < \gamma\sigma_{k+1}'$ for all $k$.
\end{lem}

The proofs are based around the following idea: check if $\delta_{k+1} \ge \gamma\delta_k$ or $\sigma_k \ge \gamma\sigma_{k+1}$, and discard pipes that violate the constraints.  The additional difficulty, relative to the analysis of GMM, arises from the special form that $f$ must satisfy and the need to bound the increase in $L$.  When we remove pipes in general the indifference points between subsequent pipes will no longer be powers of 2, so $f$ can no longer be defined in terms of $\valpha$.  We fix this by modifying the parameters of an offending pipe until the new breakpoint is a power of $2$.  To avoid drastic changes in $L$ or $f$, we achieve this by holding the cost of the given pipe $k$ fixed at its indifference point with either $k-1$ of $k+1$ and ``rotating'' the line $\sigma_k + \delta_kx$ around this fixed point until the other indifference point is fixed.

Analyzing the increase in $L$ caused by these procedures is the technical crux in the regularization analysis, as removing pipes can shift ``$\alpha$-mass'' in the multi-level cost onto much more expensive trees.  
We consider each pipe removal and the terms in $L$ it affects.  If $\alpha$-mass is shifted from $A_i(T_i^*)$ to $A_{i+l}(T_{i+l}^*)$, where $l = O(1)$, then the current chunk of $L$ has increased by $O(1)$.  If not, we show that the conditions requiring $l = \omega(1)$ imply there exist large terms in $L$ above $i+l$ that can absorb the increase with only an $O(1)$-factor loss.
We only charge against each $L$-term $O(1)$ times during the entire regularization, so the total increase is bounded by $O(1)$.

We summarize the consequences of the regularization procedure below:

\begin{thm} The algorithm $\mcA$ required by Theorem \ref{alg_thm} exists for a constant $c$, and the oblivious approximation ratio $R_1$ is constant.
\end{thm}

\section{Open Problems}

A number of interesting open problems remain to be solved.  First, we have only achieved an $O(1)$-ratio for the objective $R_1 = \max_{f} \E [f(T)]/f(T_f^*)$, but Goel and Estrin \cite{goel2003soc} have shown an $O(\log |\mathcal{D}|)$-approximation for the much harder objective $R_2 = \E \left[\max_{f} f(T)/f(T_f^*)\right]$, proving there exists a single tree that is \emph{simultaneously} an $O(\log |\mathcal{D}|)$-approximation for all $f \in \mathcal{F}$.  Achieving a constant for this stronger objective or showing a lower bound remains an important open question.  

Second, although our algorithm proves that an $O(1)$-approximate distribution exists, the ellipsoid algorithm tells us little about what these trees actually look like.  A combinatorial algorithm that yields insight as to the actual structure of these trees would also be of interest.  Third, we have made little attempt to optimize the constant $c$ in the approximation ratio, and the resulting value is huge due to the regularization procedure.
Shaving large factors off our bound on $R_1$ may be a simple question, and it would be particularly interesting to find an oblivious approximation algorithm that is competitive with standard SSBaB for known $f$.

\pdfbookmark[1]{\refname}{My\refname}
\bibliographystyle{alpha}
\bibliography{ssbab}

\appendix

\section{Proofs of regularization lemmas}
\label{regularization_appendix}

\newtheorem*{thm_max_capacity}{Lemma \ref{max_capacity_lem}}

\begin{thm_max_capacity} 
Given arbitrary $\valpha$, we can find $\valpha'$ such that the corresponding $f'$,$L'$, $\delta'$, $\sigma'$ satisfy $f(x) \le f'(x)$, $L' \le 2L$, and $\frac{\sigma'_{K-1}}{\delta'_{K-1}} \le D$, where $K$ is the number of pipes, and $D$ is the total demand rounded up to a power of 2.
\end{thm_max_capacity}

\begin{proof}
Let $k$ be the first pipe such that $\frac{\sigma_k}{\delta_k} \ge D$.  Note $k >0$ since $\frac{\sigma_0}{\delta_0} = 0$.  Remove all pipes above $k$.  
Now we modify the parameters of pipe $k$ to satisfy the desired constraint.
Increase $\delta_k$, while decreasing $\sigma_k$ so as to hold $\sigma_k + \delta_k2^{p(k-1)}$ fixed, until $\frac{\sigma_k}{\delta_k} = D$.  Geometrically, we are rotating the line $y = \sigma_k + \delta_kx$ counter-clockwise around the point $(2^{p(k-1)},\sigma_k + \delta_k2^{p(k-1)})$.  Let $\delta_k'$, $\sigma_k'$ be the new parameters for pipe $k$.  Let $f'$ be the new cost function formed by modifying pipe $k$ and removing pipes $k+1, \ldots, K-1$ and $L'$ the associated multi-level cost.

\begin{description}
\item{\emph{Claim:}}
The function $f'(x)$ is concave, and $f(x) \le f'(x)$ for all $x$.

Initially $\delta_k < \delta_{k-1}$ and $\sigma_k > \sigma_{k-1}$, and we continuously decrease $\sigma_k$ while increasing $\delta_k$.  We know $\sigma_{k-1} + \delta_{k-1}2^{p(k-1)} = \sigma_k' + \delta_k'2^{p(k-1)}$, so if we decrease $\sigma_k'$ to $\sigma_{k-1}$ the modified pipe $k$ will match pipe $k-1$.  However, we have that $\frac{\sigma_{k-1}}{\delta_{k-1}} < D = \frac{\sigma_k'}{\delta_k'}$, so we stop before reaching that point.  Therefore $\sigma_k' > \sigma_{k-1}$ and $\delta_k' < \delta_{k-1}$, which implies $f'(x)$ is concave since the switchover between pipes $k-1$ and $k$ is unchanged.
We only increased the rate of growth for $x \ge 2^{p(k-1)}$, so $f'(x) \ge f(x)$ for all $x$.

\item{\emph{Claim:}} The new multi-level cost $L'$ is at most $2L$.

There is a term $\alpha_{p(j)}$ for each changeover between pipes as well as the implicit breakpoint at $D$ when $f$ levels off.
Increasing $\delta_k$ and removing pipes $k+1, \ldots, K-1$ so that pipe $k$ is used all the way to $D$ corresponds in $L$ to pushing $\alpha$-mass from the terms $\alpha_{p(k-1)}A_{p(k-1)}(T_{p(k-1)}^*) + \cdots + \alpha_{p(K-1)}A_{p(K-1)}(T_{p(K-1)}^*)$ onto the term $\delta_k'A_{\log D}(T_{\log D}^*)$ because $p'(k) = \log D$.  

By the definition of $\sigma_k'$ and $\delta_k'$ and Lemma \ref{alphatodelta_lem} we have 
\[
\delta_k' = \frac{\sigma_k'}{D} = \sum_{j<k} \alpha_{p'(j)}'\frac{2^{p'(j)}}{D}
\]
The terms $\alpha_{p(0)}, \ldots \alpha_{p(k-2)}$ are unchanged, and $\alpha_{p(k-1)}$ drops due the decreased difference between $\delta_{k-1}$ and $\delta_k$.  There are no non-zero $\alpha_i'$ between $p(k-1)$ and $\log D$.  This gives us 
\[
\delta_k' = \sum_{j<k} \alpha_{p'(j)}'\frac{2^{p'(j)}}{D} \le \sum_{j<k} \alpha_{p(j)}\frac{2^{p(j)}}{D}
\]

Next we use Lemma \ref{L_term_lem} to relate $\frac{2^{p(j)}}{D}A_{p(j)}(T_{p(j)}^*)$ and $A_{\log D}(T_{\log D}^*)$:
\[
\delta_k'A_{\log D}(T_{\log D}^*) \le \sum_{j<k} \alpha_{p(j)}\frac{2^{p(j)}}{D}A_{\log D}(T_{\log D}^*)
\le  \sum_{j<k} \alpha_{p(j)}A_{p(j)}(T_{p(j)}^*) \le L
\]
Finally, $L' = \sum_{j<k} \alpha_{p(j)}'A_{p(j)}(T_{p(j)}^*) + \delta_k'A_{\log D}(T_{\log D}^*) \le 2L$.
\end{description}
\end{proof}

\newtheorem*{thm_delta}{Lemma \ref{delta_lem}}

\begin{thm_delta} 
Given $\valpha$ satisfying $\frac{\sigma_{K-1}}{\delta_{K-1}} \le D$, we can find $\valpha'$ 
such that the corresponding $f'$,$L'$, $\delta'$, $\sigma'$ satisfy
$f(x) \le 3f'(x)$, $L' = O(L)$, $\frac{\sigma'_{K-1}}{\delta'_{K-1}} \le D$, and $\delta'_{k+1} < \gamma\delta'_{k}$ for all $k$.
\end{thm_delta}

\begin{proof}
We repeat the following two steps until $\delta_{k+1} < \gamma\delta_k$ for all $k$.
\begin{description}
\item{1. \emph{Deletion Step:}} 
The basic idea here is the same as that used by GMM Lemma 3.2 \cite{guha2001cfa} to satisfy the constraints on the $\delta$'s: whenever a pipe violates the constraint $\delta_{k+1} \ge \gamma\delta_k$, we remove the pipe.

Let $k$ be the smallest index such that $\delta_{k+1} \ge \gamma\delta_k$, and let $l$ be the smallest 
integer such that $\delta_{k+l} < \frac{\gamma}{3}\delta_k$.  If such an $l$ exists, then remove 
pipes $k+1,\ldots,k+l-1$, and change $f(x)$ in the interval $[2^{p(k)},2^{p(k+l-1)}]$ by using the cheaper of pipe $k$ and 
$k+l$.  If no such $l$ exists then remove all pipes above $k$, and replace them with pipe $k$.
Note that this does not break the condition set in Lemma \ref{max_capacity_lem}.

\item{2. \emph{Rotation Step:}}
Pipes $k$ and $k+l$ now have equal cost at some point $g$, but $g$ may not be a power of $2$, in which case $f(x)$ is no longer in the form $\sum_i \alpha_iA_i(x)$, and $\valpha'$ is no longer defined.  

We want to modify the pipes to change $g$ while not affecting $L$ or $f$ too much.
As in Lemma \ref{max_capacity_lem}, we hold the cost of pipe $k$ fixed when routing $2^{p(k-1)}$ flow (where we switch from $k-1$ to $k$),
and reduce $\delta_k$ until pipes $k$ and $k+l$ meet at the next power of 2, increasing $\sigma_k$ to maintain $k$'s cost at $2^{p(k-1)}$.
This corresponds to rotating the line $y= \sigma_k + \delta_kx$ clockwise around the point $(2^{p(k-1)},\sigma_k + \delta_k2^{p(k-1)})$.  Let $\delta_k'$ and $\sigma_k'$ be the new parameters for pipe $k$.  
Note that $f'(x)$ now has the proper structure again, and $\valpha'$ and $L'$ are well-defined. We never increase $\sigma_0$ above $0$ since we hold this point fixed when adjusting pipe $0$.
\end{description}

\begin{figure}[htbp] 
 \centering
 \includegraphics[width=4in]{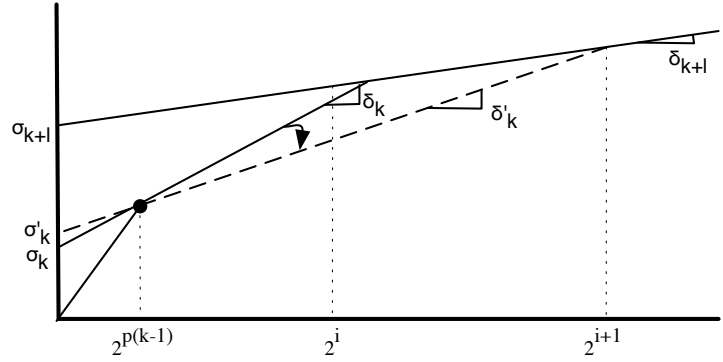} 
 \caption{To ensure the indifference point between pipes $k$ and $k+l$ is a power of 2 we ``rotate'' pipe $k$ around it's starting point until it meets $k+l$ at a power of 2.   }
 \label{rotation_fig}
\end{figure}

First, we bound the change to $\delta_k$ in the rotation step.
This allows us to prove that the constraints on the $\delta$'s are satisfied, and $f(x)$ decreases by at most an $O(1)$-factor.

\begin{description}

\item{\emph{Claim:}}
After rotation $\delta_k' \ge \frac{\delta_k}{3}$.

Before adjustment, we are indifferent between $k$ and $k+l$ at $(g,y_k)$ 
where $y_k = \sigma_k + \delta_kg = \sigma_{k+l} + \delta_{k+l}g$. 
The difference in costs between $k$ and $k+l$ at $2^{p(k-1)}$ flow remains unchanged  
because we hold the cost of pipe $k$ fixed at $2^{p(k-1)}$.
Let $x_k = g - 2^{p(k-1)}$, the distance after $2^{p(k-1)}$ at which their costs are equal.
Before rotation, the pipes' costs approach each other at a rate of $\delta_k - \delta_{k+l}$.
If we reduce $\delta_k$ by a factor of $3$, then $\frac{\delta_k}{3} - \delta_{k+l} \le \frac{1}{3}(\delta_k - \delta_{k+l})$, 
so it takes at least $3x_k$ for pipe k to grow from $\sigma_k+\delta_k2^{p(k-1)}$ to $y_k$, during which pipe $k+l$'s cost only increases, so pipe $k$ does not surpass $k+l$ until after $2^{p(k-1)} + 3x_k$.

The original pipe $k$ met pipe $k+1$ (now removed) at some point $2^{p(k)} \ge 2^{p(k-1)+1}$ before meeting $k+l$ at $g$.
Therefore $g \ge 2^{p(k-1)+1}$, which implies $x_k = g - 2^{p(k-1)} \ge \frac{g}{2}$.
After reducing $\delta_k$ to $\frac{\delta_k}{3}$, pipes $k$ and $k+l$ now meet after $2^{p(k-1)} + 3x_k = g + 2x_k \ge 2g$.
There must be a power of $2$ between $g$ and $2g$, and we reduce $\delta_k$ only until we hit the next power of $2$, so $\delta_k' \ge \frac{\delta_k}{3}$.

\item{\emph{Claim:}} When the procedure is finished $\delta_{k+1}' < \gamma\delta_k'$ for all $k$.

By the choice of $l$, $\delta_{k+l} < \frac{\gamma}{3}\delta_k \le \gamma\delta_k'$, using the previous claim.  
Further $\delta_k' < \delta_k < \gamma\delta_{k-1}$, so no previously-satisfied constraints are broken.  We renumber the pipes, and repeat the process for the next constraint violation.  When we are done, all the remaining pipes will satisfy $\delta_{k+1}' < \gamma\delta_k'$.

\item{\emph{Claim:}} For all $x$, $f(x) \le 3f'(x)$.

Note that removing pipes $k+1,\ldots,k+l-1$ only changes $f$ in the interval $(2^{p(k-1)},2^{p(k+l-1)})$, and we only remove or adjust pipes in this interval once.  Initially, removing pipes can only increase $f(x)$, but then we reduce $\delta_k$ by a factor of at most 3, which may decrease $f(x)$ by a factor of at most 3.
\end{description}

Now, we must bound the potential increase in $L$.  
To avoid confusion due to relabeling indexes after removing pipes, we change notation slightly.  Suppose the procedure completes after $K'$ iterations.  
Let $\alpha'_{p'(0)}, \ldots, \alpha'_{p'(K'-1)}$ be the final non-zero $\alpha$'s, and $\alpha_{p(0)}, \ldots, \alpha_{p(K-1)}$ the original $\alpha$'s.  
For $0 \le k \le K'-1$ let $\alpha_{p(s_k)},\ldots,\alpha_{p(s_{k+1}-1)}$ be the $L$-terms affected by the $k$th iteration of the procedure: either they are removed and merged into $\alpha'_{p'(k)}$ or $\alpha'_{p'(k)} = \alpha_{p(s_k)}$ if the constraint is already satisfied.
We need to analyze how mass is shifted between terms in $L$.  Define $L_k = \sum_{i=s_k}^{s_{k+1}-1} \alpha_{p(i)}A_{p(i)}(T^*_{p(i)})$, the portion of $L$ that round $k$ affects.

Consider round $k$ in which we remove old pipes $s_k+1, \ldots, s_{k+1}-1$ and adjust $\delta'_k$.  The old $\delta_{s_{k+1}}$ becomes $\delta_{k+1}'$.  Rotating $\delta_{k}'$ increases $\alpha'_{p'(k-1)}$ because $\alpha'_{p'(k-1)} = \delta'_{k-1} - \delta'_{k}$ but reduces the total $\alpha$-mass above $p'(k-1)$ because $\delta'_{k} = \sum_{j \ge k} \alpha'_{p(j)}$, decreasing $L$.  
The remaining $\alpha$-mass on $\alpha_{p(s_k)}A_{p(s_k)}(T_{p(s_k)}^*), \ldots$, $\alpha_{p(s_{k+1}-1)}A_{p(s_{k+1}-1)}(T_{p(s_{k+1}-1)}^*)$ merges into $\alpha'_{p'(k)}A_{p'(k)}(T_{p'(k)}^*)$ where $p'(k)$ is somewhere between $p(s_k)$ and $p(s_{k+1})$.
If mass from some $\alpha_{p(i)}$ moves down to $\alpha'_{p'(k)}$ where $p'(k) < p(i)$, then we can ignore it, as it will only reduce $L$.  If it moves up, then we will charge the increase to some higher term in $L$.

Let $c_\delta < \frac{\gamma}{3}$ be some small constant.  There are $2$ cases to consider: either $\delta_{s_{k+1}} \ge c_\delta\delta'_{k}$ or $\delta_{s_{k+1}} < c_\delta\delta'_{k}$.  
\begin{description}
\item{\emph{Case 1}}: $c_\delta\delta_{k}' > \delta_{s_{k+1}} = \delta'_{k+1}$.

Intuitively, this means there is a big drop between $\delta_{s_{k+1}-1} \ge \frac{\gamma}{3}\delta'_{k}$ and $\delta_{s_{k+1}} < c_\delta\delta'_{k}$, 
so $\alpha_{p(s_{k+1}-1)}$ must be fairly large: $\alpha_{p(s_{k+1}-1)} = \delta_{s_{k+1}-1} - \delta_{s_{k+1}} \ge (\frac{\gamma}{3} -c_\delta)\delta_{k}'$.  We will charge any increase in $L$ this iteration to the term $\alpha_{p(s_{k+1}-1)}A_{p(s_{k+1}-1)}(T_{p(s_{k+1}-1)}^*)$.
Note that we are always in this case when we remove the last pipe because we can view the last pipe as intersecting a dummy pipe with $\delta = 0$ at $D$.

In order to bound $\alpha'_{p'(k)}A_{p'(k)}(T^*_{p'(k)})$ by $\alpha_{p(s_{k+1}-1)}A_{p(s_{k+1}-1)}(T_{p(s_{k+1}-1)}^*)$
we must show that $p(s_{k+1}-1) \ge p'(k)$.
Note $2^{p'(k)}$ is the cost at which the new, rotated pipe $k$ surpasses the old pipe $s_{k+1}$.  
New pipe $k$ intersects pipe $s_{k+1}-1$ before $s_{k+1}$, and $\delta'_k > \delta_{s_{k+1}-1}$,
so pipes $k$ and $s_{k+1}$ meet before $s_{k+1}-1$ and $s_{k+1}$ do.  
Therefore $g \le 2^{p(s_{k+1}-1)}$,
and when we reduce $\delta_k'$ to fix the breakpoint we never need to raise $g$ beyond $2^{p(s_{k+1}-1)}$ before hitting a power of $2$.  Therefore

\begin{align*}
\alpha_{p(s_{k+1}-1)}A_{p(s_{k+1}-1)}(T^*_{p(s_{k+1}-1)}) &\ge \left(\frac{\gamma}{3}-c_\delta\right)\delta_{k}'A_{p(s_{k+1}-1)}(T^*_{p(s_{k+1}-1)}) && \text{(by assumption)}\\
&\ge \left(\frac{\gamma}{3}-c_\delta\right)\alpha'_{p'(k)}A_{p(s_{k+1}-1)}(T^*_{p(s_{k+1}-1)}) && \text{(using $\delta_k' = \sum_{j \ge k} \alpha_{p(j)}'$)}\\
&\ge \left(\frac{\gamma}{3}-c_\delta\right)\alpha'_{p'(k)}A_{p'(k)}(T^*_{p'(k)})
\end{align*}
We can charge the increase in $\alpha'_{p'(k)}$ to $\alpha_{p(s_{k+1}-1)}$ in the current chunk $L_k$, with a loss of $\left(\frac{\gamma}{3}-c_\delta\right)^{-1} = \frac{3}{\gamma-3c_\delta}$, and this charge can only occur once for each $L_k$.

\item{\emph{Case 2}:} $c_\delta\delta_{k}' \le \delta_{s_{k+1}}$.

In this case there is no large collection of mass that we can easily guarantee is above $p'(k)$ in the current interval, but we do know there must be a lot of mass somewhere above $p(s_{k+1}-1)$ because $\delta_{s_{k+1}}$ is large.  
The $\alpha$-mass $\alpha_{p(s_{k+1})}+ \ldots + \alpha_{p(s_{k+2}-1)} = \delta_{s_{k+1}} - \delta_{s_{k+2}}$ is ``used'' in the next iteration and contributes $L_{k+1}$ to $L$.
We know $\gamma\delta_{s_{k+1}} = \gamma\delta'_{k+1}  > \delta_{s_{k+2}}$, which implies $\sum_{i=s_{k+1}}^{s_{k+2}-1} \alpha_{p(i)} = \delta_{s_{k+1}} - \delta_{s_{k+2}} > (1 - \gamma)\delta_{s_{k+1}}$.
Now we can bound the increase
\begin{align*}
\alpha'_{p'(k)}A_{p'(k)}(T^*_{p'(k)}) & \le \left(\delta'_{k} - \delta_{s_{k+1}}\right)A_{p'(k)}(T^*_{p'(k)}) &&
(\alpha'_{p'(k)} = \delta'_k - \delta_{s_{k+1}}) \\
&\le \left(\frac{1}{c_\delta} - 1\right)\delta_{s_{k+1}}A_{p'(k)}(T^*_{p'(k)}) && 
\text{(by assumption)}\\
&\le \left(\frac{1}{c_\delta} - 1\right)\left(\frac{1}{1-\gamma}\sum_{i=s_{k+1}}^{s_{k+2}-1} \alpha_{p(i)}\right)A_{p'(k)}(T^*_{p'(k)}) && 
\text{(shown above)}\\
&\le \frac{1-c_\delta}{c_\delta(1-\gamma)}
\sum_{i=s_{k+1}}^{s_{k+2}-1} \alpha_{p(i)}A_{p(i)}(T^*_{p(i)}) && 
(p'(k) < p(s_{k+1}) \le p(i) \\
&= \frac{1-c_\delta}{c_\delta(1-\gamma)}L_{k+1} &&
\end{align*}
Therefore we can charge the increase in $L$ due to iteration $k$ to the portion $L_{k+1}$ used in the next iteration. 
\end{description}

For a particular segment $L_k$ of $L$, the $k-1$th iteration may been bounded by $\frac{1-c_\delta}{c_\delta(1-\gamma)}$ increase in $L_k$, and the $k$th iteration may charge against a $\frac{3}{\gamma-3c_\delta}$ increase.  Each type of charge can occur at most once per chunk.  Therefore the total increase in each piece, and hence the total increase in $L = \sum_k L_k$ is
\[
\frac{1-c_\delta}{c_\delta(1-\gamma)} + \frac{3}{\gamma-3c_\delta}
\]
This completes the proof.
\end{proof}

\newtheorem*{thm_sigma}{Lemma \ref{sigmalem}}

\begin{thm_sigma} 
Given $\valpha$ satisfying $\frac{\sigma_{K-1}}{\delta_{K-1}} \le D$ and $\delta_{k+1} < \gamma\delta_{k}$,
we can find $\valpha'$ such that such the corresponding $f'$,$L'$, $\delta'$, $\sigma'$ satisfy
$f(x) \le \frac{5}{2}f'(x)$, $L'  = O(L)$, $\delta'_{k+1} < \gamma\delta'_{k}$, and $\sigma_k' < \gamma\sigma_{k+1}'$ for all $k$.
\end{thm_sigma}

\begin{proof}  
The proof follows Lemma \ref{delta_lem} but moves backwards through the pipes rather than forwards.
\begin{description}
\item{1. \emph{Deletion Step:}}
Let $k$ be the highest index such that $\sigma_{k-1} \ge \gamma\sigma_k$, and $l>1$ the smallest integer such that $\sigma_{k-l} < \frac{2\gamma}{5}\sigma_k$.  Such an $l$ must exist because $\sigma_0 = 0$.  
Remove pipes $k-l+1, \ldots, k-1$, and replace them with the cheaper of pipes $k-l$ and $k$.

\item{2. \emph{Rotation Step:}}
As in Lemma \ref{delta_lem}, $f(x)$ may no longer be a linear combination of terms $A_i(x)$ because the new indifference point may not be a power of 2.  We use a similar procedure as before to remedy this.
Hold pipe $k$'s cost for $2^{p(k)}$ flow fixed, and reduce $\sigma_k$ while increasing $\delta_k$ to maintain the invariant until $k$ and $k-l$ meet at a power of 2.  Geometrically we are rotating $y = \sigma_k + \delta_kx$ counter-clockwise around $(2^{p(k)},\sigma_k + \delta_k2^{p(k)})$.  Let $\sigma'_k$, $\delta'_k$ be the new parameters. Note that $\valpha'$ and $L'$ are now well-defined.
\end{description}

First, we analyze the change to $\sigma_k$ and $\delta_k$ required by the rotation step and use this result to prove the constraints on both the $\sigma$'s and $\delta$'s are satisfied at the end without changing $f(x)$ too much.

\begin{description}
\item{\emph{Claim:}}
After rotation $\sigma'_k \ge \frac{2}{5}\sigma_k$, and $\delta'_k \le \frac{8}{5}\delta_k$.

Suppose the unmodified pipe $k$ and $k-l$ meet at $g = \frac{\sigma_k - \sigma_{k-l}}{\delta_{k-l} - \delta_k}$.  We will bound the adjustment required to guarantee they meet before $\frac{g}{2}$.  
Reduce $\sigma_k$ to $\frac{2}{5}\sigma_k = \sigma'_k$.
The modified pipe $k$ has the same cost as the old at $2^{p(k)}$.
If $k$ is the final pipe then from Lemma \ref{max_capacity_lem} we know D = $2^{p(k)} \ge \frac{\sigma_k}{\delta_k}$.  
Otherwise, pipe $k$ costs the same as $k+1$ at $2^{p(k)}$, so we have that $2^{p(k)} = \frac{\sigma_{k+1}-\sigma_k}{\delta_k - \delta_{k+1}} \ge \frac{\sigma_k}{\delta_k}$, using $\gamma\sigma_{k+1} > \sigma_k$ (the constraint fixed in the previous iteration).  In either case $\delta_k2^{p(k)} \ge \sigma_k$.
Now,
\begin{align*}
\sigma_k + \delta_k2^{p(k)} &= \frac{2}{5}\sigma_k + \delta'_k2^{p(k)}  \\
\Rightarrow \delta'_k2^{p(k)} = \frac{3}{5}\sigma_k + \delta_k2^{p(k)} &\le \left(1+\frac{3}{5}\right)\delta_k2^{p(k)} 
\Rightarrow \delta'_k \le \frac{8}{5}\delta_k
\end{align*}

The constraints on the $\delta$s were satisfied before removing pipe $k-1$, so $\delta_{k-l} > \frac{1}{\gamma^2}\delta_k$.  This implies 
\[
\frac{\delta_k}{\delta_{k-l}-\delta_k} \le \frac{\delta_{k}}{\frac{1}{\gamma^2}\delta_{k}- \delta_{k}} \le \frac{1}{4-1} = \frac{1}{3}
\]
using $\gamma < \frac{1}{2}$.
We combine this with the bound on $\delta'_k$ to bound the change in $\delta_{k-l} - \delta_k'$: 
\begin{align*}
\delta_{k-l} - \delta'_k \ge (\delta_{k-l} - \delta_k) - \frac{3}{5}\delta_k =& (\delta_{k-l}-\delta_k)\left(1-\frac{3}{5}\frac{\delta_k}{\delta_{k-l}-\delta_k}\right) \\
\ge& (\delta_{k-l}-\delta_k)\left(1-\frac{3}{5}\cdot\frac{1}{3}\right) = \frac{4}{5}(\delta_{k-l}-\delta_k)
\end{align*}

Now we have enough information to bound the new switchover point.
\[
\frac{\sigma'_k -\sigma_{k-l}}{\delta_{k-l} - \delta'_k} 
= \frac{\frac{2}{5}\sigma_k -\sigma_{k-l}}{\delta_{k-l} - \delta'_k} 
\le \frac{\frac{2}{5}(\sigma_k -\sigma_{k-l})}{\delta_{k-l} - \delta'_k}
\le \frac{\frac{2}{5}(\sigma_k -\sigma_{k-l})}{\frac{4}{5}(\delta_{k-l} - \delta_k)}
= \frac{1}{2}\frac{\sigma_k -\sigma_{k-l}}{\delta_{k-l} - \delta_k} 
= \frac{1}{2}g
\]
There must be a power of 2 between $\frac{g}{2}$ and $g$, so we need to reduce $\sigma_k$ by at most a factor of $\frac{2}{5}$.
Finally, note that pipes 0 and 1 meet no sooner than 1, and $k>1$ since it is always true that $\gamma\sigma_1 > \sigma_0 = 0$.  Therefore $g > 1$, and hence the new changeover point is at least 1, so we do not need to worry about a term $A_{-1}$.

\item{\emph{Claim:}} When the procedure finishes $\delta'_{k+1} < \gamma\delta_k'$ and $\sigma_k' < \gamma\sigma_{k+1}'$ for all $k$.

We chose $l$ such that $\sigma_{k-l} < \frac{2\gamma}{5}\sigma_k$, so $\sigma_{k-l} < \gamma\sigma_k'$.  Before starting, we had $\gamma^2\delta_{k-l} > \gamma\delta_{k-1} > \delta_k$, and $\gamma < \frac{1}{2}$, which implies $\delta_k' \le \frac{8}{5}\delta_k < \frac{8}{5}\gamma^2\delta_{k-l} < \gamma\delta_{k-l}$.  Note that the rotation step does not break any previously-satisfied constraints on larger $k$'s.

\item{\emph{Claim:}} For all $x$, $f(x) \le \frac{5}{2}f'(x)$.

Only 1 round affects the interval $(2^{p(k-l-1)},2^{p(k)})$.  Removing pipes only increases $f(x)$, and if we adjust $\sigma_k$, then it decreases by a factor of at most $\frac{2}{5}$, while $\delta_k$ increases, so $f'(x) \ge \frac{2}{5}f(x)$.
\end{description}

Now we analyze the increase in $L$.  
First, unlike in Lemma \ref{delta_lem}, the rotation step works against us, and we need to bound the increase.
\begin{description}
\item{\emph{Claim:}}
Rotation only increases $L$ by an $O(1)$-factor.

When adjusting pipe $k$, we increase $\delta_k$ without changing $\delta_{k+1}$, which increases $\alpha_{p(k)}$.  We have that $\alpha_{p(k)} \ge (1-\gamma)\delta_k$, and $\delta'_k \le \frac{8}{5}\delta_k$, so
\[
\alpha'_{p(k)} = \delta'_k - \delta_{k+1} \le (\delta_k - \delta_{k+1})\left(1 + \frac{3}{5}\frac{\delta_k}{\delta_k - \delta_{k+1}}\right) \le \alpha_{p(k)}\left(1 + \frac{3}{5}\frac{\delta_k}{\delta_k(1-\gamma)}\right) = \frac{8-5\gamma}{5(1-\gamma)}\alpha_{p(k)}
\]
causing $L$ to increase by at most $\frac{8-5\gamma}{5(1-\gamma)}$.
\end{description}

Second, we need to bound the increase in $L$ caused by removing pipes.
Let $K'$ be the number of iterations and final pipes and $\alpha'_{p'(0)},\ldots,\alpha_{p'(K'-1)}'$ the resulting non-zero $\alpha$'s.  Iteration $k$, for $1 \le k \le K'$, deletes pipes $s_{k+1}+1, \ldots, s_{k}-1$ which removes $\alpha_{p(s_{k+1})},\ldots, \alpha_{p(s_k-1)}$.  
Let $L_k = \sum_{i=s_{k+1}}^{s_k-1} \alpha_{p(i)}A_{p(i)}(T^*_{p(i)})$ be the amount these contribute to $L$.
Since it moves backwards through pipes the indices of new pipes are not fixed yet, but as labeled at the end, round $k$ ensures $\sigma'_{j} < \gamma\sigma'_{j+1}$ and creates a term $\alpha'_{p'(j)}$ where $j = K' - k$.
  
The rotation step reduces both $\alpha'_{p'(j)}$ and $p'(j)$ which can only help in this step,
and we have already bounded the increase in $\alpha'_{p'(j+1)}$ due to rotation, so we assume that no rotation is needed.  
This implies $\alpha'_{p'(j)} = \delta_{s_{k+1}} - \delta_{s_k} = \sum_{i=s_{k+1}}^{s_k-1} \alpha_{p(i)}$.  
As in Lemma \ref{delta_lem} we need to ensure that too much $\alpha$-mass does not move too high.

Let $c_{\sigma} < \frac{2\gamma}{5}$ be a small constant.  We need to consider two cases again: either $\sigma_{s_{k+1}} < c_{\sigma}\sigma'_{j+1}$ or $\sigma_{s_{k+1}} \ge c_{\sigma}\sigma'_{j+1}$.

\begin{description}
\item{\emph{Case 1}:} $\sigma_{s_{k+1}} < c_{\sigma}\sigma'_{j+1}$.

Intuitively, this means $\sigma_{s_{k+1}+1}$ is much larger than $\sigma_{s_{k+1}}$ because $\sigma_{s_{k+1}+1} \ge \frac{2\gamma}{5}\sigma'_{j+1}$, 
so by the time pipe $s_{k+1}$ catches up with pipe $s_{k+1}+1$ or any later pipe, it has already covered an $O(1)$-fraction of the distance to $2^{p'(j)}$.  Therefore, pushing mass from up to $A_{p'(j)}(T_{p'(j)}^*)$ increases $L$ by only a constant factor.

We bound $2^{p'(j)}$ by bounding the cost to which pipe $s_{k+1}$ must grow before switching pipes.
Before removal the old pipe $s_k-1$ crossed the new $j+1$ at $2^{p(s_k-1)} = \frac{\sigma_{j+1}' - \sigma_{s_k-1}}{\delta_{s_k-1}-\delta'_{j+1}} \le \frac{\sigma_{j+1}'}{\frac{1}{\gamma}\delta_{j+1}' - \delta_{j+1}'} \le \frac{\sigma_{j+1}'}{\delta_{j+1}'}$, 
so $\sigma_{j+1}' + \delta_{j+1}'2^{p(s_k-1)} \le 2\sigma_{j+1}'$.  Pipe $s_{k+1}$'s cost increases faster than $s_k-1$'s and surpasses $s_k$'s cost before $2^{p(s_k-1)}$.  Therefore $\sigma_{j+1}' + \delta_{j+1}'g \le 2\sigma_{j+1}'$.  

We know $\sigma_{s_{k+1}+1} \ge \frac{2\gamma}{5}\sigma_{j+1}'$ or else it would not have been removed.  When $s_{k+1}$ intersects $s_{k+1}+1$ at $2^{p(s_{k+1})}$ it has grown from $\sigma_{s_{k+1}}$ to at least $\sigma_{s_{k+1}+1}$ and therefore has covered at least
\[
\frac{\sigma_{s_{k+1}+1} - \sigma_{s_{k+1}}}{2\sigma_{j+1}'} \ge \frac{\frac{2\gamma}{5}\sigma_{j+1}' - c_\sigma\sigma_{j+1}'}{2\sigma_{j+1}'} = \frac{2\gamma-5c_\sigma}{10}
\]
fraction of the distance to the indifference point between $s_{k+1}+1$ and $s_{k}$.
Therefore
\[
2^{p(s_{k+1})} \ge \frac{2\gamma-5c_\sigma}{10}2^{p'(j)} \Rightarrow
A_{p'(j)}(T_{p'(j)}^*) \le \frac{10}{2\gamma-5c_{\sigma}}A_{p(s_{k+1})}(T_{p(s_{k+1})}^*)
\]
Every other affected $\alpha_{p(i)}$ is pushed up less than $\alpha_{p(s_{k+1})}$, so
\begin{align*}
\alpha_{j}'A_{p'(j)}(T^*_{p'(j)}) &= \sum_{i=s_{k+1}}^{s_k-1} \alpha_{p(i)}A_{p'(j)}(T^*_{p'(j)}) \\
 &\le \sum_{i=s_{k+1}}^{s_k-1} \alpha_{p(i)}\left(\frac{10}{2\gamma-5c_{\sigma}}A_{p(i)}(T^*_{p(i)}) \right) 
 = \frac{10}{2\gamma-5c_{\sigma}}L_k
\end{align*}

\item{\emph{Case 2}:} $\sigma_{s_{k+1}} \ge c_{\sigma}\sigma'_{j+1}$.

In this case pipes $s_{k+1}$ and $s_{k+1}+1$ may meet very early, and $A_{p'(j)}(T_{p'(j)}^*)$ could be much bigger than $A_{p(s_{k+1})}(T_{p(s_{k+1})}^*)$.  Note that we are never in this case when $\sigma_{s_{k+1}} = 0$.  We have that
\begin{gather*}
\sigma_{j+1}' + \delta_{j+1}2^{p'(j)} = \sigma_{s_{k+1}} + \delta_{s_{k+1}}2^{p'(j)} \\
\Rightarrow
\alpha'_{p'(j)} = \delta_{s_{k+1}} - \delta_{j+1}' = \frac{\sigma_{j+1}' - \sigma_{s_{k+1}}}{2^{p'(j)}} \le \left(\frac{1}{c_\sigma}-1\right) \frac{\sigma_{s_{k+1}}}{2^{p'(j)}}
\end{gather*}
After the next round---which we know occurs because $\sigma_{s_{k+1}} \neq 0$---
$\sigma_{s_{k+2}}$ will be the pipe preceding $\sigma_{s_{k+1}}$ (which is $\sigma_j'$).
Using $\sigma_{s_{k+2}} < \gamma\sigma_{s_{k+1}}$, it is easy to see that
$
\sigma_{s_{k+1}} < \frac{\sigma_{s_{k+1}}-\sigma_{s_{k+2}}}{1-\gamma}
$
and from the formula for $\sigma_{s_{k+2}}$ we have
$
\sigma_{s_{k+1}} - \sigma_{s_{k+2}} = \sum_{i=s_{k+2}}^{s_{k+1}-1} \alpha_{p(i)}2^{p(i)}
$

Combining the previous inequalities,
\begin{align*}
\alpha'_{p'(j)} \le& \left(\frac{1-c_\sigma}{c_\sigma}\right)\frac{\sigma_{s_{k+1}}}{2^{p'(j)}} 
\le  \left(\frac{1-c_\sigma}{c_\sigma}\right)\left(\frac{\sigma_{s_{k+1}}-\sigma_{s_{k+2}}}{1-\gamma}\right)\frac{1}{2^{p'(j)}} \\
\le& \frac{1-c_\sigma}{c_\sigma(1-\gamma)}\sum_{i=s_{k+2}}^{s_{k+1}-1} \alpha_{p(i)}2^{p(i)-p'(j)}\\
\end{align*}
Now we can apply Lemma \ref{L_term_lem} to finish the bound:
\begin{align*}
\alpha'_{p'(j)}A_{p'(j)}(T_{p'(j)}^*) \le& \frac{1-c_\sigma}{c_\sigma(1-\gamma)}\sum_{i=s_{k+2}}^{s_{k+1}-1} \alpha_{p(i)}2^{p(i)-p'(j)}A_{p'(j)}(T_{p'(j)}^*) \\
\le& \frac{1-c_\sigma}{c_\sigma(1-\gamma)}\sum_{i=s_{k+2}}^{s_{k+1}-1} \alpha_{p(i)}A_{p(i)}(T_{p(i)}^*) = \frac{1-c_\sigma}{c_\sigma(1-\gamma)}L_{k+1}
\end{align*}
Therefore we can charge the increase in $L_k$ this iteration to $L_{k+1}$ used in the next iteration.
\end{description}

For a particular chunk $L_k$ of $L$, round $k$'s increase may be bounded by a $\frac{10}{2\gamma-5c_\sigma}$-factor increase and round $k-1$ may be bounded by a $\frac{1-c_\sigma}{c_\sigma(1-\gamma)}$-factor increase.  Each charge only occurs once.
The rotation step adds another factor of $\frac{8-5\gamma}{5(1-\gamma)}$ on top of this.
Therefore, the total growth of $L$ is at most
\[
\frac{8-5\gamma}{5(1-\gamma)} \left(\frac{1-c_\sigma}{c_\sigma(1-\gamma)} + \frac{10}{2\gamma-5c_\sigma}\right)
\]
\end{proof}
\end{document}